\documentclass[final,12pt]{colt2022} 

\title[Streaming Algorithms for Ellipsoidal Approximation of Convex Polytopes]{Streaming Algorithms for Ellipsoidal Approximation of Convex Polytopes}
\usepackage{times}

\coltauthor{\Name{Yury Makarychev} \Email{yury@ttic.edu}\\
 \Name{Naren Sarayu Manoj} \Email{nsm@ttic.edu}\\
 \Name{Max Ovsiankin} \Email{maxov@ttic.edu}\\
 \addr Toyota Technological Institute Chicago}

\usepackage[T1]{fontenc}
\usepackage[breakable, theorems, skins]{tcolorbox}
\tcbset{enhanced}
\allowdisplaybreaks

\usepackage{xcolor}
\usepackage{mathtools}
\usepackage{enumitem}
\usepackage{bbm}
\usepackage{semantic}
\usepackage{listings}
\usepackage{courier}
\usepackage{tikz}
\usetikzlibrary{calc}
\usepackage{color}
\usepackage{fancyhdr}
\usepackage{lastpage}

\definecolor{codegreen}{rgb}{0,0.6,0}
\definecolor{codegray}{rgb}{0.5,0.5,0.5}
\definecolor{codepurple}{rgb}{0.58,0,0.82}
\definecolor{backcolour}{rgb}{0.95,0.95,0.92}
 
\lstdefinestyle{mystyle}{
    backgroundcolor=\color{backcolour},   
    commentstyle=\color{codegreen},
    keywordstyle=\color{magenta},
    numberstyle=\tiny\color{codegray},
    stringstyle=\color{codepurple},
    basicstyle=\footnotesize,
    breakatwhitespace=false,         
    breaklines=true,                 
    captionpos=b,                    
    keepspaces=true,                 
    numbers=left,                    
    numbersep=5pt,                  
    showspaces=false,                
    showstringspaces=false,
    showtabs=false,                  
    tabsize=2,
    basicstyle=\footnotesize\ttfamily
}

\usepackage{hyperref}

\newcommand{\paren}[1]{\left(#1\right)}

\newtheorem{claim}[theorem]{Claim}

\newtheorem{problem}[theorem]{Problem}

\newcommand{\logv}[1]{\log\paren{#1}}
\newcommand{\diag}[1]{\mathsf{diag}\paren{#1}}

\renewcommand{\vec}[1]{\overrightarrow{#1}}

\makeatletter
\newcommand{\vo}{\vec{o}\@ifnextchar{^}{\,}{}}
\makeatother

\usepackage{tcolorbox}

\usepackage{algorithm}
\usepackage{algorithmic}

\numberwithin{equation}{section}

\def\eps{\varepsilon}
\def\epsilon{\varepsilon}

\def\eps{\epsilon}

\def\phi{\varphi}
\def\cal{\mathcal}

\newcommand{\R}{{\mathbb R}}

\newcommand{\conv}[1]{\mathsf{conv}\inparen{#1}}

\usepackage{nicefrac}

\let\nfrac=\nicefrac

\newcommand{\abs}[1]{\ensuremath{\left\lvert #1 \right\rvert}}

\newcommand{\norm}[1]{\ensuremath{\left\lVert #1 \right\rVert}}

\newcommand{\fnorm}[1]{\ensuremath{\left\lVert #1 \right\rVert_{F}}}

\newcommand{\ip}[1]{\left\langle #1 \right\rangle}

\newcommand{\suchthat}{{\;\; : \;\;}}

\newcommand{\argmax}{\mathrm{argmax}}

\newcommand{\inparen}[1]{\left(#1\right)}             
\newcommand{\inbraces}[1]{\left\{#1\right\}}           
\newcommand{\insquare}[1]{\left[#1\right]}             

\def\eps{\varepsilon}
\def\epsilon{\varepsilon}

\def\eps{\epsilon}

\def\phi{\varphi}
\def\cal{\mathcal}

\newcommand{\cA}{\mathcal{A}}

\newcommand{\cE}{\mathcal{E}}

\usepackage{multirow}
\usepackage{soul}

\usepackage{thmtools}
\usepackage{thm-restate}

\usepackage{parskip}
\usepackage{multirow}

\newcommand{\mybox}{}

\usepackage[
backend=biber,
style=trad-alpha,
maxalphanames=8,
sorting=nty
]{biblatex}
\usepackage{cleveref}

\usepackage{ulem}

\newcommand{\sigmamax}[1]{\sigma_{\max, #1}}

\newcommand{\detv}[1]{\mathsf{det}\inparen{#1}}
\newcommand{\magicnumber}[1]{3+14d\log #1 + 8d}

\newcommand{\astar}{J^{*}}

\newcommand{\potfun}[1]{\Phi_{#1}}

\definecolor{ForestGreen}{RGB}{34, 139, 34}
\newcommand{\greenbox}[1]{
    \vspace{0.15cm}
    { \small
    \begin{tcolorbox}[enhanced, colback=ForestGreen!12, parbox=false]
    #1
    \end{tcolorbox}
    }
}

\newcommand{\edits}[1]{\textcolor{red}{#1}}

\definecolor{forest}{RGB}{0,155,85}

\newcommand{\fixout}{\bgroup\markoverwith{\textcolor{forest}{\rule[0.5ex]{2pt}{0.4pt}}}\ULon}

\newcommand{\kappaol}{\kappa^{\mathsf{OL}}}

\begin{document}

\maketitle

\begin{abstract}%
We give efficient deterministic one-pass streaming algorithms for finding an ellipsoidal approximation of a symmetric convex polytope. The algorithms are near-optimal in that their approximation factors differ from that of the optimal offline solution only by a factor sub-logarithmic in the aspect ratio of the polytope.
\end{abstract}


\section{Introduction}

Let $X$ be a centrally symmetric convex body in $\mathbb{R}^d$. We say that an ellipsoid $\cE$ is an \textit{$\alpha$-ellipsoidal approximation for $X$} if $\nfrac{\cE}{\alpha} \subseteq X \subseteq \cE$ (where $\alpha \geq 1$). Calculating ellipsoidal approximations has applications to problems in machine learning and data science, including sampling and volume estimation (see, e.g., \cite{cousins15} and \cite{he21}), obstacle collision detection in robotics (see \cite{boyd97}), differential privacy (see \cite{nikolov13}), and online learning (see \cite{li2019stochastic}). Further, the ellipsoid
$\cE$ provides a very succinct approximate representation of $X$ -- $\cE$ can be stored using only 
$\binom{d+1}{2} = O(d^2)$ floats, while the exact representation of $X$ may be arbitrarily large.

John's theorem \cite{john1948} states that the minimum-volume outer ellipsoid, called \emph{John's Ellipsoid of $X$}, is a $\sqrt{d}$-ellipsoidal approximation when $X$ is symmetric -- that is, when $X = -X$. Similarly, John's theorem implies that the maximum-volume inner ellipsoid also yields a $\sqrt{d}$-ellipsoidal approximation when $X$ is symmetric. Furthermore, the approximation factor of $\sqrt{d}$ given by John's Ellipsoid cannot be improved in the worst case (e.g. for the hypercube or cross-polytope). 
John's result can be made algorithmic: Cohen, Cousins, Lee, and Yang designed an highly efficient algorithm for computing the maximum volume interior ellipsoid when $X$ is a symmetric polytope defined by linear constraints, giving a $\sim \sqrt{d}$ approximation  in the offline setting \cite{pmlr-v99-cohen19a}.

\noindent\textbf{Streaming Algorithm.} A natural follow-up question is whether a similar approximation guarantee exists for convex polytopes given in the \textit{streaming setting}. Specifically, suppose that we are given points or constraints defining a symmetric convex polytope one-at-a-time. Our goal is to design an algorithm that computes an $\alpha$-ellipsoidal approximation to the polytope and uses as little memory as possible. Such an algorithm will be useful in a memory-constrained environment. For instance, consider a streaming data summarization task in which a user wishes to obtain an approximation of a dataset that is too large to fit in memory. By computing a good ellipsoidal approximation, the user can summarize the dataset in only $\binom{d+1}{2}$ floating point numbers, whereas to store all vertices of the polytope $X$ we may need $nd$ floats. To our knowledge, existing solutions (such as that of \cite{pmlr-v99-cohen19a}) require the entire dataset to be stored in memory and therefore cannot be applied to a streaming setting.

\paragraph{Related Works and Applications} The problem of calculating an ellipsoidal approximation to a convex body has been well-studied; see \cite{todd16} for an overview of the area. The recent paper of \cite{pmlr-v99-cohen19a} presents an $\widetilde{O}\inparen{nd^2}$-time algorithm for computing a $\sim \sqrt{d}$-approximation for $X$ when $X$ is specified by symmetric linear constraints. 

The problem of approximating a \textit{non-symmetric} convex hull with an ellipsoid was introduced in \cite{mukhopadhyay2010approximate}.
The authors present a greedy algorithm for this problem and show that its approximation factor is unbounded for every $d \geq 2$. However, they do not provide any upper bounds on the approximation ratio.

Efficiently calculating ellipsoidal approximations has implications to the problem of estimating the volume of a convex body. For instance, it is known (see \cite{cousins15}) that if a convex body $X$ satisfies $B_2^d \subseteq X \subseteq R \cdot B_2^d$, then its volume can be approximated in time $\widetilde{O}\inparen{\max\inbraces{d^2R^2, d^3}}$ using a procedure known as \emph{Gaussian cooling}. Observe that computing an $\alpha$-ellipsoidal approximation yields a linear transformation $T$ that transforms a convex body $X$ into a position such that $B_2^d \subseteq TX \subseteq \alpha \cdot B_2^d$. Hence, an efficient algorithm to compute an $\alpha$-ellipsoidal approximation for small $\alpha$ yields an efficient algorithm for estimating the volume of a convex polytope.

Ellipsoidal approximations have also been used as \emph{exploration bases} in some online optimization problems. For instance, the work of \cite{li2019stochastic} considers a stochastic linear optimization setting with adversarial corruptions. Here, a learner observes noisy evaluations of a linear function with the goal of maximizing this function over a convex constraint set. The algorithm used in \cite{li2019stochastic} uses an ellipsoidal rounding of the constraint set to construct an exploration basis. The learner then uses this exploration basis to sample actions during its exploration phases. The approximation factor of the rounding plays a role in the expected regret of the algorithm. Thus, a good, efficiently computable ellipsoidal approximation can be used as a black box to obtain a more efficient, lower-regret algorithm for this setting.

It is highly desirable to give a streaming algorithm for problems in compute-constrained scenarios. For instance, the work of \cite{boyd97} studies a problem in which a robot must estimate the distance between a collection of obstacles and itself. At a high level, their workflow involves computing an ellipsoidal approximation of the convex hull of the set of obstacles. The robot then calculates the distance between itself and this ellipsoidal approximation. Now, the robot could be operating using a microcontroller or some other device with limited computing capabilities and at the same time, the amount of data may be very large. Therefore, we need a memory-efficient algorithm for this setting.

Finally, a streaming algorithm to compute ellipsoidal approximations would yield an algorithm that can adapt to certain changes in a dataset over time while maintaining a consistent approximation guarantee. For instance, suppose that the robot in the setting of \cite{boyd97} has only a limited sight distance. As the robot moves, it acquires knowledge of new obstacles. A streaming algorithm for ellipsoidal approximations would allow the robot to quickly update its summary of the set of obstacles as it moves.

\subsection{Main Result} 

In this paper, we study the following formalization of the \emph{Ellipsoidal Approximation Problem}, stated below.
\begin{problem}[Ellipsoidal Approximation Problem]
\label{problem:ellipsoidal_approx_general}
Given a symmetric convex body $X$, find an ellipsoid $\cE$ so that $\nfrac{\cE}{\alpha} \subseteq X \subseteq \cE$ for some $\alpha \geq 1$. The goal is to find an approximation with a small value of $\alpha$. We say that $\cE$ is an $\alpha$-approximation to $X$. 
\end{problem}

We study the Ellipsoidal Approximation Problem in the following streaming model. We assume that points $x_1, \dots, x_n \in X$ arrive one-by-one. After the algorithm receives point $x_t$, it must output an ellipsoid $\cE_t$ centered at the origin such that all points $x_1,\dots, x_t$ lie in $\cE_t$. The \emph{approximation factor} of the algorithm is the smallest $\alpha$ such that 
$\nfrac{\cE_n}{\alpha} \subseteq \conv{\inbraces{\pm x_1, \dots, \pm x_n}} \subseteq \cE_n$ for every input sequence. The value of $n$ might not be known to the algorithm beforehand.

\begin{problem}[Formal Problem Statement]
\label{problem:main_problem}
We observe a stream of points $x_1, \dots, x_n$ where points arrive one-by-one and $n$ might not be known beforehand. Upon receiving point $x_t$, find an ellipsoid $\cE_t$ such that:
\begin{itemize}
    \item For all $t \in \{1,\dots,n\}$, we have $\conv{\inbraces{\pm x_1, \dots, \pm x_t}} \subseteq \cE_t$;
    \item At the end of the stream, we have $\nfrac{\cE_n}{\alpha} \subseteq \conv{\inbraces{\pm x_1, \dots, \pm x_n}} \subseteq \cE_n$.
\end{itemize}
\end{problem}

The primary results of our work are algorithms presented in Theorem~\ref{thm:main_result_informal}.

\begin{theorem}[Main Result]
\label{thm:main_result_informal}
Let $X =\conv{\inbraces{\pm x_1, \dots, \pm x_n}}$. Assume that $X$ contains a ball of radius $r$ and is contained in a ball of radius $R$.
\begin{enumerate}
    \item There is a streaming algorithm that given $r$ and a stream of points $x_1,\dots,x_n$ provides a solution for Problem~\ref{problem:main_problem} with $\alpha = O(\sqrt{d\logv{\nfrac{R}{r}+1}})$.
    \item There is a streaming algorithm that given $\nfrac{R}{r}$ (but not $r$ and $R$) and a stream of points $x_1,\dots,x_n$ provides a solution for Problem~\ref{problem:main_problem} with $\alpha = O(\sqrt{d\logv{\nfrac{R}{r}+1}})$.
\end{enumerate}
The algorithms run in time $\widetilde{O}\inparen{nd^2}$ and store $O(d^2)$ floating-point numbers.
\end{theorem}

The ratio \(\nfrac{R}{r}\) gives an upper bound to the \textit{aspect ratio} \(\kappa(X)\), a quantity we formally define in Definition~\ref{def:aspect_ratio}.
We present and analyze the algorithm for item 1 in Theorem~\ref{thm:simple_alg_arb_ellipsoid} and the one for item 2 in  Theorem~\ref{thm:greedy_approx}. 
 Our results provide a nearly optimal approximation, since $\sqrt{d\logv{\nfrac{R}{r} +1}}$ is worse than $\sqrt{d}$ (the best possible factor for the offline setting) by only a factor of $O\inparen{\sqrt{\logv{\nfrac{R}{r} +1}}}$. Our algorithms store only $O(d^2)$ floats -- this amount of memory is necessary to represent an ellipsoid in $\mathbb{R}^d$.
We note that the algorithm from item 1 is similar to the algorithm for the non-symmetric case of the problem presented in~\cite{mukhopadhyay2010approximate}.

\textbf{Hyperplane representation of a polytope.} Our algorithms also work in the setting where instead of receiving points $x_i$, we receive hyperplanes (or, more precisely, slabs) $\{y:|\langle x_i, y\rangle| \leq 1\}$ defining body $X$. We describe this setting in Appendix~\ref{section:mvie}.

\textbf{Applications to volume estimation.} Our algorithms directly yield concrete results for the problem of estimating the volume of a symmetric convex polytope. In particular, composing our ellipsoidal approximation procedures with the Gaussian cooling algorithm of \cite{cousins15} yields an $\widetilde{O}\inparen{nd^2+d^3\logv{\nfrac{R}{r}} +d^3}$-time algorithm for approximating the volume of a symmetric convex polytope. This guarantee is comparable to that obtained by using the algorithm of \cite{pmlr-v99-cohen19a} as a preprocessing step prior to Gaussian cooling. Furthermore, this improves upon the best guarantee for this problem suggested by \cite{he21}, which is $\widetilde{O}\inparen{nd^{3.2}\mathrm{polylog}\inparen{\nfrac{R}{r}}}$ time (\cite{he21} assumes the KLS hyperplane conjecture but does not require that $X$ be symmetric).

\textbf{Limitations to approximating John's ellipsoid.} We note that one could try to obtain a good ellipsoidal approximation $\cE$ for $X$ in the streaming model by finding a $\beta$-approximation for John's ellipsoid. Such an ellipsoid $\cE$ would provide a $\beta \sqrt{d}$ approximation for $X$. However, we give a lower bound on the approximability of John's ellipsoid for a natural class of streaming algorithms. We show that none of the algorithms in this class can find a better than $\sqrt{d}$-approximation for John's ellipsoid. Thus, the approach of approximating $X$ via approximating John's ellipsoid potentially may only yield an $O(\sqrt{d} \cdot \sqrt{d}) = O(d)$ approximation.

\paragraph{Independent and concurrent work on convex hull approximation.}
Independent of and concurrent to our work,
Woodruff and Yasuda~\cite{woodruff2022high} consider a problem of maintaining a coreset for $\ell_\infty$ subspace embedding: given a stream of points, choose a small subset of them so that the symmetric convex hull of this subset is a good approximation to the convex hull of the original points.
This problem and the one we study are closely related. In both of them, we need to maintain an approximation for the symmetric convex hull of a stream of points. However, Woodruff and Yasuda approximate the convex hull with a convex hull of a small coreset, while we approximate the convex hull with an ellipsoid. The approximation guarantee of \cite{woodruff2022high} is $O(\sqrt{d\logv{n\kappaol(X)}}$, where $\kappaol(X)$ is the online condition number.
It has the same dependence on $d$ as our guarantee; $\kappaol(X)$ is closely related to but different from the parameter $\kappa(X)$  we use.
Also, by applying John's Theorem, the authors of \cite{woodruff2022high} get an $O(d\sqrt{\logv{n\kappaol(X)}})$ ellipsoidal approximation to the convex hull. Our approximation guarantee is better than this one by a factor of $\sim \sqrt{d}$ and does not depend on $n$.

\paragraph{A Proof Sketch of the Main Result.} We now give a proof outline of Theorem \ref{thm:main_result_informal}.

Our first algorithm, Algorithm~\ref{alg:greedy_maxev}, maintains an ellipsoid $\cE_t$ covering $\conv{\inbraces{\pm x_1, \dots, \pm x_t}}$. Initially, $\cE_0$ is simply a ball of radius $r$. Upon receiving point $x_{t+1}$, the algorithm updates $\cE_t$ to $\cE_{t+1}$ by computing the minimal volume ellipsoid containing $\cE_t$ and $\pm x_{t+1}$. 

We need to prove that this algorithm achieves an $\alpha = O(\sqrt{d\logv{\nfrac{R}{r} +1}})$ approximation. By construction, $\cE_n$ contains $X$. Now we need to prove that $\cE_n \subseteq \alpha \cdot X$. We first prove a  different statement. We compare ellipsoid $\cE_n$ not to $X$ but rather to an ellipsoid $\cE^{*}$ that contains $X$. We show that $\cE_n \subseteq \alpha' \cE^{*}$ for every ellipsoid $\cE^{*}$ such that (i) $\cE^{*}$ contains $X$ and (ii) $\cE^{*}$ has aspect ratio (the ratio of its longest to shortest semi-axis) at most $\nfrac{R}{r}$. 
To this end, we define a potential function $\Phi$ (see Definition~\ref{def:potential}) with the following properties:
\begin{itemize}
    \item $\cE_n \subseteq \alpha' \cE^{*}$ for $\alpha'=O(\sqrt{\Phi})$ (see Lemma~\ref{lemma:phi_control_max_sv})
    \item Initially, $\Phi$ is $O(d\logv{\nfrac{R}{r} +1})$ (see Lemma~\ref{lemma:simple_alg_initial}).
    \item The value of the potential function is non-increasing over time (see Lemma~\ref{lemma:invariants_basic}).
\end{itemize}
These properties imply that 
$\cE_n \subseteq \alpha' \cE^{*}$ for $\alpha'=O(\sqrt{d\logv{\nfrac{R}{r} +1}})$.
Then we prove in Theorem~\ref{thm:simul_ellipsoid_approx}, that this implies that $\cE_n\subseteq \alpha X$ with $\alpha = \sqrt{2}\alpha'$, as required.

For this algorithm to perform well, it must know $r$ or a reasonable estimate for $r$.
However, if we do not have any estimate on $r$, the performance of the algorithm may be arbitrarily bad. 
On the technical level, the challenge is that the initial value of potential $\Phi$ may be arbitrarily large.

Algorithm~\ref{alg:general_greedy_maxev} does not need to know $r$ but instead needs to have an estimate $\xi$ for the aspect ratio of $X$. Algorithm~\ref{alg:general_greedy_maxev} updates $\cE_i$ in two steps: first it performs the update step from Algorithm~\ref{alg:greedy_maxev} and then ensures that the aspect ratio of the obtained ellipsoid is roughly at most $\xi$ (if it is more than that, it expands the semiaxes of the ellipsoid appropriately). The analysis of Algorithm~\ref{alg:general_greedy_maxev} is based on that of Algorithm~\ref{alg:greedy_maxev} but is substantially more complex. We use a pair of potential function $S$ and $P$ and keep track of their evolutions over time.

\paragraph{Outline} The rest of our paper is organized as follows. In Section~\ref{section:preliminaries_notation}, we present definitions and notation used in this paper. In Section~\ref{section:natural_attempt}, we describe the first algorithm from Theorem~\ref{alg:greedy_maxev}. The algorithm itself is very simple but its analysis is insightful and  captures the core technical ideas used later. In Section \ref{section:revised_algorithm}, we present the second algorithm from Theorem~\ref{alg:general_greedy_maxev}. 
The analyses of Theorems~\ref{alg:greedy_maxev} and \ref{alg:general_greedy_maxev} relies on the fact that every centrally symmetric convex body is well approximated by an intersection of ellipsoids with bounded aspect ratio. We prove this fact in Appendix \ref{section:ellipsoidal_approximations}. In Appendix~\ref{section:tracking_mvoe}, we present our lower bound on the approximability of John's ellipsoid. Finally, in Appendix~\ref{section:mvie}, we discuss the equivalence between the problem we study and an alternate formulation wherein we receive linear constraints one-at-a-time instead of points.

\section{Preliminaries and Notation}
\label{section:preliminaries_notation}

\paragraph{Notation} Consider a sequence of points $\inbraces{x_1,\dots, x_n} \subset {\mathbb R}^d$. We denote the symmetric convex hull of the first $t$ points by $X_t=\conv{\inbraces{\pm x_1, \dots, \pm x_t}}$
and the symmetric convex hull of all points $\{\pm x_i\}$ by $X = X_n$. 
We denote the standard Euclidean norm of a vector $v$ by $\norm{v}$ and the Frobenius norm of a matrix $A$ by $\fnorm{A}$. We denote the singular values of a matrix $A \in \R^{d\times d}$ by $\sigma_1(A),\dots, \sigma_d(A)$.
Let $\sigma_{\max}(A)$ and $\sigma_{\min}(A)$ be the largest and smallest  singular values of $A$, respectively. We say that $X$ is centrally symmetric if $X=-X$.

Denote the $\ell_p$-unit ball by $B_p^d = \{x\in \R^d: \|x\|_p \leq 1\}$.   Given a set $S \subseteq \mathbb{R}^d$, its polar is $S^\circ \coloneqq \{ y \in \mathbb{R}^d \colon \sup_{x \in S} | \langle x, y \rangle | \leq 1 \}$. We use natural logarithms unless otherwise specified. 

In this paper, we will work extensively with ellipsoids. We will always assume that all ellipsoids and balls we consider are centered at the origin; we will not explicitly state that. We use the following representation of ellipsoids.
For a non-singular matrix $A\in\R^{d\times d}$, let $\cE_A \coloneqq \inbraces{x \suchthat \norm{Ax}\le 1}$. In other words, matrix $A$ defines a bijective map  of $\cE_A$ to the unit ball $B_2^d$.
Every ellipsoid (centered at the origin) has such a representation.
We note that this representation is not unique as matrices $A$ and $MA$ define the same ellipsoid if matrix $M$ is orthogonal (since $\|Av\| =\|MAv\|$ for every vector $v$). Now consider the singular value decomposition of  $A$: $A = U\Sigma^{-1} V^T$ (it will be convenient for us to write $\Sigma^{-1}$ instead of standard $\Sigma$ in the decomposition). The diagonal entries of $\Sigma$ are exactly the semi-axes of $\cE_A$. As mentioned above, matrices $U\Sigma^{-1} V^T$ and $U'\Sigma^{-1} V^T$ define the same ellipsoid for any orthogonal \(U' \in \mathbb{R}^{d \times d}\); in particular, every ellipsoid can be represented by a matrix of the form $A=\Sigma^{-1} V^T$.

Our goal is to design an algorithm for Problem \ref{problem:main_problem} that achieves a good approximation $\alpha$ and at the same time uses as little memory as possible. To understand what value of $\alpha$ is achievable in the offline case, recall John's Theorem.

\mybox{\begin{theorem}[John's Theorem, \cite{john1948}]
\label{thm:johns}
I. Let $X$ be a centrally symmetric convex body. Consider the minimum volume ellipsoid $\cE$ containing $X$. Then we have $\nfrac{\cE}{\sqrt{d}} \subseteq X \subseteq \cE$.

II. There exists a centrally symmetric body $X$ in ${\mathbb R}^d$ (e.g. hypercube  $B_\infty^d$ and cross-polytope $B_1^d$) such that
there is no ellipsoid $\cE$ that approximates $X$ within a factor 
of $\alpha < \sqrt{d}$: $\nfrac{\cE}{\sqrt{d}} \subseteq X \subseteq \cE$.
\end{theorem}}


In this work, we consider a natural class of one-pass streaming algorithms that we call \emph{monotonic algorithms}.

\begin{definition}[Monotonic Algorithm]
We call an algorithm for Problem \ref{problem:main_problem} \emph{monotonic} if it produces a sequence of ellipsoids $\cE_t$ satisfying $X_t \subseteq \cE_t$ and $\cE_{t} \subseteq \cE_{t+1}$ for all timestamps $t$. 
\end{definition}

Monotonic algorithms have the advantage that once they decide
that a certain point $x$ belongs to ellipsoid $\cE_t$, they commit to this decision: all consecutive ellipsoids $\cE_{t+1}, \cE_{t+2}, \dots$ also contain point $x$.
The approximation factor of our algorithms depend sublogarithmically on the \textit{aspect ratio} of the convex body $X$. 

\begin{definition}[Aspect Ratio]\label{def:aspect_ratio}
Consider a centrally symmetric convex body $X$. Let $r$ be the radius of the largest ball $r\cdot B_2^d$ contained in $X$ and $R$ be the radius of the smallest ball $R\cdot B_2^d$ that contains $X$.
Then the aspect ratio of $X$ is written as $\kappa(X) = \nfrac{R}{r}$.
\end{definition}

Logarithmic dependences on the aspect ratio have previously appeared for algorithms on convex bodies; for example, the algorithms in \cite{he21} for rounding and computing the volume of a convex body have a runtime that depends on \(\log \kappa(X)\).
We also recall the condition number of a matrix:

\begin{definition}[Condition Number of a Matrix]
The condition number $\kappa(A)$ of a symmetric nonsingular matrix $A$ is the ratio of its largest to smallest singular values:
$\kappa(A) = \nfrac{\sigma_{\max}(A)}{\sigma_{\min}(A)}$.
\end{definition}
The notions of aspect ratio of a convex body and condition number of a matrix are closely related. It is immediate that the aspect ratio of an ellipsoid $\cE$ equals the ratio of its longest to shortest semi-axes. Consequently, $\kappa(\cE_A) = \kappa(A)$.

\section{Scale-Dependent Algorithm for Ellipsoid Approximation}
\label{section:natural_attempt}

In this section, we present and analyze a simple algorithm for Problem \ref{problem:main_problem}; see Algorithm~\ref{alg:greedy_maxev}. This algorithm must be given a radius $r$ such that $r \cdot B_2^d \subseteq X$.
The approximation guarantee of the algorithm linearly depends on 
$\sqrt{\logv{\nfrac{R}{r}}}$ where $R = \max_t \|x_t\|$ is the radius of the smallest ball that contains $X$.
\begin{algorithm}
\caption{Streaming Ellipsoidal Approximation -- Scale-Dependent Algorithm\label{alg:greedy_maxev}}
\begin{algorithmic}[1]
    \STATE \textbf{Input}: A stream of points $x_1, \dots, x_n$ and a value $r$ such that:
    \begin{align*}
        r \cdot B_2^d \subset X = \conv{\inbraces{\pm x_1, \dots, \pm x_n}}
    \end{align*}
    \STATE \textbf{Output}: Ellipsoid $\cE_n$ that covers $X$.
    \STATE Initialize $\cE_0 = r \cdot B_2^d$.
    \FOR{$t = 1, \ldots, n$}
        \STATE Read point $x_t$ from the stream.
        \STATE Let $\cE_t$ be the ellipsoid of smallest volume (centered at $0$) that contains both $\cE_{t-1}$ and $x_t$.
        \label{line:simple_update_rule}
    \ENDFOR
    \STATE \textbf{Output}: $\cE_n$.
\end{algorithmic}
\end{algorithm}

The key line in the algorithm is Line \ref{line:simple_update_rule} which updates \(\cE_t\) to contain \(x_t\); we  refer to it as the ``update rule.''
Claim \ref{claim:minimum_volume_ellipsoid_update}, which we prove in Appendix~\ref{sec:rank-one-update}, shows how to compute the update.
\mybox{\begin{claim}
\label{claim:minimum_volume_ellipsoid_update}
Given a matrix $A_{t-1}$ for $\cE_{t-1}$, the updated matrix 
$A_{t}$ for $\cE_{t}$ can be updated using the following formula: $A_t = \widehat{A} A_{t-1}$, where
\begin{equation}
\label{eq:A-hat}
\widehat{A} = \inparen{I - \inparen{1 - \frac{1}{\norm{A_{t-1}x_t}}}\inparen{\frac{\inparen{A_{t-1}x_t}\inparen{A_{t-1}x_t}^T}{\norm{A_{t-1}x_t}^2}}}.
\end{equation}
\end{claim}}

We first show that ${\cal E}_n$ provides a good approximation to every ellipsoid \({\cal E}^{*}\) with aspect ratio $\nfrac{R}{r}$ containing $X$.
\mybox{\begin{theorem}
\label{thm:simple_alg_arb_ellipsoid}
Let ${\cal E}^{*} \supseteq X$ be an ellipsoid containing $X$ with aspect ratio at most $\nfrac{R}{r}$.
Then the output of Algorithm \ref{alg:greedy_maxev} $\cE_n$ satisfies:
\begin{align*}
    {\cal E}_n \subseteq \alpha \cdot \cE^{*}
\end{align*}
where $\alpha = O\inparen{\sqrt{d\inparen{\logv{\nfrac{R}{r}}+1}}}$. Algorithm \ref{alg:greedy_maxev} runs in time $O(nd^2)$ and stores at most $O(d^2)$ floating point numbers.
\end{theorem}}

As stated, Theorem~\ref{thm:simple_alg_arb_ellipsoid} does not say that $\cE_n$ provides an
$\alpha$ approximation for $X$. However, the statement of \Cref{thm:simple_alg_arb_ellipsoid} holds simultaneously for all ellipsoids 
${\cal E}^{*}$ whose aspect ratio is at most that of \(X\).
As the following theorem shows, this is sufficient to get the desired result that $\cE_n \subseteq \alpha\sqrt{2}\cdot X$.

\mybox{\begin{theorem}
\label{thm:simul_ellipsoid_approx}
Consider a centrally symmetric convex body $X$ and an ellipsoid $\cE$. 
Assume that every ellipsoid $\cE^{*}$ that satisfies properties (i) and (ii)
\begin{center}
(i) $\cE^{*}$ contains $X$ \quad and \quad (ii) $\cE^{*}$ has aspect ratio of at most $\kappa(X)$
\end{center}
also contains $\cE$. Then $\cE \subseteq \sqrt{2} \cdot X$.
\end{theorem}}
We prove \Cref{thm:simul_ellipsoid_approx} in \Cref{section:ellipsoidal_approximations}.
Combining these theorems, we have \Cref{thm:simple_alg}.

\mybox{\begin{theorem}
\label{thm:simple_alg}
Algorithm~\ref{alg:greedy_maxev} gets an $\alpha = O\inparen{\sqrt{d\logv{\nfrac{R}{r} + 1}}}$ approximation:
$
    \nfrac{\cE_n}{\alpha} \subseteq X \subseteq \cE_n
$.
\end{theorem}}

We focus now on proving \Cref{thm:simple_alg_arb_ellipsoid}.
\begin{proof}[Proof of \Cref{thm:simple_alg_arb_ellipsoid}]
\paragraph{Runtime and Memory Complexity} The key observation is that Algorithm \ref{alg:greedy_maxev} only stores the matrix $A_t$ representing the ellipsoid $\cE_t$ between iterations and does not require additional memory within an iteration. Hence, the memory complexity of Algorithm \ref{alg:greedy_maxev} is $O(d^2)$. Next, observe that computing the update rule as per Claim \ref{claim:minimum_volume_ellipsoid_update} requires only three matrix-vector products, which takes time $O(d^2)$. It immediately follows that the runtime of Algorithm \ref{alg:greedy_maxev} is $O(nd^2)$, as desired.

\paragraph{Correctness} We assume without loss of generality that for all $t\geq 1$, $\|A_{t-1}x_t\| > 1$, since Algorithm \ref{alg:greedy_maxev} ignores all points $x_t$ with $\|A_{t-1}x_t\| \leq 1$. In particular, when Algorithm \ref{alg:greedy_maxev} encounters such a point, it simply lets $A_t = A_{t-1}$.
Additionally, we assume that the shortest semi-axis of $\cE^{*}$ is at most $R$.
If not, we prove the statement for $\cE^{**} = R\cdot B_2^d$ (which is contained in $\cE^{*}$ by our assumption)  and get $\cE_n \subseteq  \alpha \cE^{**} \subseteq \alpha \cE^{*}$, as required.

Let $\astar$ be a matrix that defines the ellipsoid $\cE^{*}$: $\cE^{*} = \inbraces{x \suchthat \norm{\astar x} \le 1}$. Since all points $x_t$ lie in $X \subseteq \cE^{*}$, we have $\norm{\astar x_t} = \norm{\astar (-x_t)} \le 1$.
Since the shortest semi-axis of $\cE^{*}$ is at most $R$, we have \(\sigma_{\max}(\astar) \geq \nfrac{1}{R}\).
Further, as \(r \cdot B_2^d \subseteq X \subseteq \cE^{*}\), we have \(\sigma_{\max}(\astar) \leq \nfrac{1}{r}\).

Now we  define a potential function $\potfun{\astar}(\cE_t)$. We will show that the value of this function does not increase over time. 
We will then upper bound the approximation factor $\alpha$ in terms of $\potfun{\astar}(\cE_n)$.

\begin{definition}[Potential Function $\potfun{\astar}\inparen{\cdot}$]
\label{def:potential}
We define the potential function $\potfun{\astar}(A_t)$ as:
\begin{align*}
S_t &\coloneqq \fnorm{\astar \cdot A_{t}^{-1}}^2 = \sum_{i=1}^d \sigma_i(\astar A_t)^2\\
P_t &\coloneqq 2\log\detv{\astar\cdot A_{t}^{-1}} =\log \prod_{i=1}^d \sigma_i(\astar A_t)^2 \\
\potfun{\astar}(A_t) &\coloneqq S_t - P_t = \fnorm{\astar \cdot A_t^{-1}}^2 - 2\log \detv{\astar \cdot A_t^{-1}}
\end{align*}
\end{definition}


Let us see how we use $\potfun{\astar}(\cdot)$ to upper bound the singular values of $\astar A_n^{-1}$ and the approximation factor $\alpha$. 
\mybox{\begin{lemma}
\label{lemma:phi_control_max_sv}
\begin{enumerate}
    \item $\sigma_{\mathrm{max}}(\astar \cdot A_{n}^{-1})^2 \le \frac{e}{e-1}\cdot \potfun{\astar}(A_n)$ 
    \item $\cE_n \subseteq \alpha \cE^{*}$, where $\alpha = \sigma_{\mathrm{max}}(\astar \cdot A_{n}^{-1})$.
\end{enumerate}    
\end{lemma}}
\begin{proof}
1. Let $M=\astar \cdot A_n^{-1}$. Note that $f(t):=t^2 - \log t^2 \geq \frac{e-1}{e}\cdot t^2 > 0$ for $t>0$.
Therefore,
\begin{align*}
    \potfun{\astar}(A_n) &=
    \sum_{i=1}^d \inparen{\sigma_{i}^2(M) - \log \sigma_{i}^2(M)} 
    = \sum_{i=1}^d f(\sigma_{i}(M))
    \geq f(\sigma_{\mathrm{max}}(M))
    \geq \frac{e-1}{e}\cdot\sigma_{\mathrm{max}}^2(M)
\end{align*}
as desired.

2. Consider a point $x\in \cE_n$. Then $\|A_n x\|\leq 1$. We get
$$\norm{\astar x} = \norm{\astar A_{n}^{-1} \cdot A_nx}
\le \sigma_{\max}\inparen{\astar \cdot A_{n}^{-1}} \cdot \norm{A_nx} 
\leq \alpha.$$
We conclude that $x\in \alpha \cE^{*}$. Thus, $\cE_n \subseteq \alpha \cE^{*}$.
\end{proof}

Let us describe the plan for the rest of the proof.
\begin{itemize}
\item In Lemma~\ref{lemma:invariants_basic}, we will prove that $\potfun{\astar}(A_t) \le \potfun{\astar}(A_{t-1})$ for all $t\geq 1$; that is, values $\potfun{\astar}(A_{t-1})$ are non-increasing.
\item In Lemma~\ref{lemma:simple_alg_initial}, we give an upper bound  
$\potfun{\astar}(A_{0}) \leq O(d \log (\nicefrac{R}{r} +1))$. 
\end{itemize}
From Lemma \ref{lemma:phi_control_max_sv}, we get that $\cE_n \subseteq \alpha \cdot \cE^{*}$ with
$$
\alpha \leq O(\sqrt{\potfun{\astar}(A_{n})}) \leq O(\sqrt{\potfun{\astar}(A_{0})}) \leq
O(\sqrt{d \log (\nicefrac{R}{r} +1)}).
$$
It remains to prove Lemmas~\ref{lemma:invariants_basic} and~\ref{lemma:simple_alg_initial}
mentioned above. We start with Lemma~\ref{lemma:invariants_basic}.

\mybox{\begin{lemma}
\label{lemma:invariants_basic}
Under the update rule for Algorithm \ref{alg:greedy_maxev}, we have $\potfun{\astar}(A_t) \le \potfun{\astar}(A_{t-1})$ for all $t\geq 1$.
\end{lemma}}
\begin{proof}
We first derive formulas that express $P_t$ and $S_t$ in terms of $P_{t-1}$ and $S_{t-1}$.
\mybox{\begin{lemma}
\label{lemma:invariants_precorrection}
We have,
\begin{align*}
    P_t &= 2\log\norm{A_{t-1}x_t} + P_{t-1}\\
    S_t &= S_{t-1} + \inparen{1 - \frac{1}{\norm{A_{t-1}x_t}^2}}\norm{\astar\cdot x_t}^2
\end{align*}
\end{lemma}}
\begin{proof}
By Claim~\ref{claim:minimum_volume_ellipsoid_update}, $A_t = \widehat{A} A_{t-1}$ where $\hat A$ is given by (\ref{eq:A-hat}).

\paragraph{Determinant Update}
We start by calculating $\detv{\widehat{A}}$. Note that 
$\inparen{1 - \frac{1}{\norm{A_{t-1}x_t}}}\inparen{\frac{\inparen{A_{t-1}x_t}\inparen{A_{t-1}x_t}^T}{\norm{A_{t-1}x_t}^2}}$ is a symmetric rank-1 matrix, whose only non-zero
eigenvalue equals $\inparen{1 - \frac{1}{\norm{A_{t-1}x_t}}}$.
Therefore, the spectrum of $\widehat{A}$ consists of $1$ with multiplicity \(d-1\) and $\norm{A_{t-1}x_t}^{-1}$ with multiplicity \(1\). Thus, $\detv{\widehat{A}} = \norm{A_{t-1}x_t}^{-1}$ and we have:
\begin{align*}
    \detv{\astar A_{t}^{-1}}
    = \detv{\astar A_{t-1}^{-1}\widehat{A}^{-1}}
    = \detv{\astar A_{t-1}^{-1}} \cdot \detv{\widehat{A}^{-1}}
\end{align*}
\paragraph{Frobenius Norm Update} It is well-known (see, e.g., \cite{horn1991}) that for any matrix $A$ and orthonormal matrix $V$ (with $v_1, \dots, v_d$ as its columns):
\begin{align*}
    \fnorm{A}^2 = \fnorm{AV}^2 = \sum_{i=1}^d \norm{Av_i}^2
\end{align*}
Let $V$ be a matrix consisting of the eigenvectors of $\widehat{A}$. Observe that one of these vectors must be $v_1 \coloneqq \nfrac{A_{t-1}x_t}{\norm{A_{t-1}x_t}}$; let $v_2, \dots, v_d$ denote the remaining eigenvectors, all of which have an associated eigenvalue of $1$. Now we calculate:
\begin{align*}
    S_t &= \fnorm{\astar\cdot A_{t}^{-1}}^2 = \fnorm{\astar\cdot A_{t-1}^{-1}\widehat{A}^{-1}}^2 
    = \sum_{i=1}^d \norm{\astar\cdot A_{t-1}^{-1}\widehat{A}^{-1}v_i}^2 \\
    &= \norm{A_{t-1}x_t}^2 \norm{\astar\cdot A_{t-1}^{-1}v_1}^2 + \sum_{i=2}^d \norm{\astar\cdot A_{t-1}^{-1}v_i}^2 \\
    &= \inparen{\norm{A_{t-1}x_t}^2-1} \norm{\astar\cdot A_{t-1}^{-1}v_1}^2 + \sum_{i=1}^d \norm{\astar\cdot A_{t-1}^{-1}v_i}^2 \\
    &= \inparen{\norm{A_{t-1}x_t}^2-1} \norm{\astar\cdot A_{t-1}^{-1} \cdot \frac{A_{t-1}x_t}{\norm{A_{t-1}x_t}}}^2 + \fnorm{\astar\cdot A_{t-1}^{-1}}^2 \\
    &= S_{t-1} + \inparen{1 - \frac{1}{\norm{A_{t-1}x_t}^2}}\norm{\astar\cdot x_t}^2
\end{align*}
\end{proof}

Now are ready to prove Lemma~\ref{lemma:invariant_basic_stronger}, which implies
Lemma~\ref{lemma:invariants_basic}.
In fact, it is stronger than what we need to prove Lemma \ref{lemma:invariants_basic}; however, it will be necessary in the sequel.

\mybox{\begin{lemma}
\label{lemma:invariant_basic_stronger}
For all $t\in\{1,\dots, n\}$, we have:
\begin{align*}
    S_t - S_{t-1} &\le \inparen{P_t - P_{t-1}}\norm{\astar \cdot x_t}^2
\end{align*}
\end{lemma}}
\begin{proof}
From Lemma \ref{lemma:invariants_precorrection}, we have:
\begin{align*}
    \frac{S_t - S_{t-1}}{P_t - P_{t-1}} &= \frac{\inparen{1-\nfrac{1}{\norm{A_{t-1}x_t}^2}}\norm{\astar \cdot x_t}^2}{2\logv{\norm{A_{t-1}x_t}}} \le \norm{\astar x_t}^2\quad\quad\\
\end{align*}
where the inequality follows from the fact that \(\frac{\inparen{1-\nfrac{1}{y^2}}}{2\log y} < 1\) for all \(y > 1\).
We now multiply both sides by $P_t - P_{t-1}$ to obtain \(S_t - S_{t-1} \leq (P_t - P_{t-1}) \cdot \| J^{*} x_t\|^2\).
\end{proof}
As $x_t \in \cE^{*}$, we have $\norm{\astar x_t} \le 1$. Using Lemma \ref{lemma:invariants_basic} and rearranging, we find: \[\potfun{\astar}(A_t) = S_t - P_t \leq S_{t-1} - P_{t-1} = \potfun{\astar}(A_{t-1})\]
This concludes the proof of Lemma~\ref{lemma:invariants_basic}.
\end{proof}

\mybox{\begin{lemma}
\label{lemma:simple_alg_initial}
We have $\potfun{\astar}(A_0) \leq d
\inparen{1 + 4\logv{\nfrac{R}{r}}}$.
\end{lemma}}
\begin{proof}
Recall that $R = \max_{x \in X} \norm{x}$.

Let $f(t):=t^2 - \log t^2$ as in the proof of 
Lemma~\ref{lemma:phi_control_max_sv}.
By the definition of the potential function, we have
$$
    \potfun{\astar}(A_0) = \fnorm{\astar \cdot A_0^{-1}}^2 - 2\log \detv{\astar \cdot A_0^{-1}} 
    = \sum_{i=1}^d f(\sigma_i(\astar \cdot A_0^{-1}))
$$
Now we bound the range of $\sigma_i(\astar \cdot A_0^{-1})$. Observe that we have $\cE_0 = r\cdot B_2^d$ and, accordingly, $A_0 = \nfrac{1}{r} \cdot I$. Hence, $\astar \cdot A_0^{-1} =  r \astar$ and $\sigma_i(\astar \cdot A_0^{-1}) = r \sigma_i(\astar)$.
Now recall that by assumption \(\nfrac{1}{R} \leq \sigma_{\max}(\astar) \leq \nfrac{1}{r}\), and the condition number of \(\astar\) is at most \(\nfrac{R}{r}\).
Thus \(\sigma_{\min}(\astar) \geq \frac{\sigma_{\max}(\astar)}{\nfrac{R}{r}} \geq \nfrac{r}{R^2}\), and consequently \(\nfrac{r^2}{R^2} \leq \sigma_i(\astar \cdot A_0^{-1}) \leq 1\).
Now we bound $f(\sigma_i(\astar \cdot A_0^{-1}))$. Since \(f\) is convex,
\begin{align*}
    f(t) \leq \max\inbraces{f\inparen{ \nfrac{r^2}{R^2}},f(1)} \leq 1 + 2\logv{ \nfrac{R^2}{r^2}} = 1 + 4 \logv{\nfrac{R}{r}} \quad\text{for } t\in \insquare{ \nfrac{r^2}{R^2},1}
\end{align*}
We conclude that
\begin{align*}
    \potfun{\astar}(A_0) = \sum_{i=1}^d f(\sigma_i(\astar \cdot A_0^{-1})) \leq d\inparen{1 + 4\logv{\nfrac{R}{r}}}.
\end{align*}
\vspace{-\baselineskip} 
\end{proof}
We have proven Lemmas~\ref{lemma:invariants_basic} and~\ref{lemma:simple_alg_initial}. This concludes the proof of Theorem~\ref{thm:simple_alg_arb_ellipsoid}.
\end{proof}

\section{Scale-Independent Algorithm}
\label{section:revised_algorithm}

To use Algorithm~\ref{alg:greedy_maxev}, we need to know some lower bound $r$ on the radius of the largest ball contained in $X$. The approximation guarantee of the algorithm linearly depends on $\sqrt{\logv{\nfrac{R}{r}}}$, so as long as we have a reasonable estimate on $r$, we can use Algorithm~\ref{alg:greedy_maxev}. However, if we have no prior information about the scale of $X$ and do not have any reasonable estimate $r$, we cannot use Algorithm~\ref{alg:greedy_maxev}.
In this section, we present Algorithm~\ref{alg:general_greedy_maxev} that only requires an upper bound $\xi \geq \kappa(X)$ on the aspect ratio of $X$.
\begin{algorithm}[ht]
\caption{Streaming Ellipsoidal Approximation\label{alg:general_greedy_maxev}}
\begin{algorithmic}[1]
    \STATE \textbf{Input}: A stream of points $x_1, \dots, x_n$, and an aspect ratio estimate \(\xi\)
    \STATE \textbf{Output}: Matrix $A_n$ that defines ellipsoid $\cE_n = \inbraces{x \suchthat \norm{A_nx} \le 1}$.
    \STATE Receive point $x_1$.
    \STATE Set $V_1$ to be an orthonormal matrix satisfying $V_1^T x_1 = \norm{x_1} \cdot e_1$.
    \STATE $U_1 = I$, $\Sigma_1 = \diag{\norm{x_1}, \nfrac{\norm{x_1}}{\xi}, \dots, \nfrac{\norm{x_1}}{\xi}}$
    \STATE  \(A_1 = U_1 \Sigma_1^{-1} V_1^T\) 
    \FOR{$t = 2, \ldots, n$}
        \STATE Receive point $x_t$
        \IF{$\norm{A_{t-1} x_t} > 1$}
            \STATE $a = \nfrac{A_{t-1} x_t}{\norm{A_{t-1} x_t}} \text{ and } b = A_{t-1}^T a$.
            \STATE $M_t = \max(M_{t-1}, \|x_t\|)$ \quad \textit{(that is, $M_t = \max_{1\leq i \leq t} \|x_i\|$)}
            \STATE $(U_t, (\Sigma_t')^{-1}, V_t) = \textsc{SVDRankOneUpdate}((U_{t-1}, \Sigma_{t-1}^{-1}, V_{t-1}),  - \inparen{1 - \nfrac{1}{\norm{A_{t-1} x_t}}}a, b)$
            \label{line:svd_update}
            \STATE \(\Sigma_t = \diag{\tau_{1, t}, \ldots, \tau_{d, t}}\), where $\tau_{i,t} = \max([\Sigma'_{t}]_{i i}, \nfrac{M_t}{\xi})$ for every $i \in [d]$
             \label{line:sv_correction}
        \ELSE
            \STATE $U_t = U_{t-1}, V_t = V_{t-1}, \Sigma_t = \Sigma_{t-1}$
        \ENDIF
        \STATE  \(A_{t} = U_t \Sigma_t^{-1} V_t^T\)
    \ENDFOR
    \STATE \textbf{Output}: $\cE_n = \inbraces{x \suchthat \norm{A_n x}\le 1}$
\end{algorithmic}
\end{algorithm}
We present Algorithm~\ref{alg:general_greedy_maxev} with its full implementation details.
However, conceptually the only difference between Algorithms~\ref{alg:greedy_maxev} and~\ref{alg:general_greedy_maxev} is a new ``singular value correction step'', presented on line~\ref{line:sv_correction}. As in Algorithm~\ref{alg:greedy_maxev}, we perform the basic update rule (line~\ref{line:svd_update}) -- compute the minimum volume ellipsoid $\cE'_t$ containing the current ellipsoid $\cE_{t-1}$ and the new point $x_t$.
Then, we increase the semi-axes of $\cE'_t$ (if necessary) so that all of them are at least $M_t/\xi$, where $M_t = \max(\|x_1\|,\dots,\|x_t\|)$ is the length of the longest vector we have received so far. This ensures that the matrix $A_t$ is well-conditioned. 

Formally, we compute matrix $A_t'=\widehat{A} A_{t-1}$ where $\widehat A$ is given by (\ref{eq:A-hat}), and update its singular values.
To implement this efficiently, we use a procedure 
$$(U', \Sigma', V') = \textsc{SVDRankOneUpdate}((U, \Sigma, V), y, z),$$
which gives \(U', \Sigma', V'\) as the SVD of the matrix \(U \Sigma V^T + y z^T\).
Per \cite{stange08}, this can be implemented in time $O(d^2\log^2 d)$.
As noted in Section \ref{section:preliminaries_notation}, 
the \(U_t\)s have no effect on the definition of \(\mathcal{E}_t\)s and thus will not factor into our analysis. In fact, the algorithm may discard the value of $U_t$ after it computes the SVD for $A_t$ and later use the identity matrix instead of $U_t$. However, we keep them in the algorithm so that the exposition is closer to that of Algorithm~\ref{alg:greedy_maxev}.

The remainder of this section is devoted to proving Theorem~\ref{thm:greedy_approx}.

\mybox{\begin{theorem}
\label{thm:greedy_approx}
Algorithm \ref{alg:general_greedy_maxev} outputs an ellipsoid $\cE_n$ satisfying:
\begin{align*}
    \frac{\cE_n}{\sqrt{6 + 28d \log \xi + 16d}
    } \subseteq X \subseteq \cE_n
\end{align*}
assuming $\xi \ge \kappa(X)$. 
Moreover, Algorithm \ref{alg:general_greedy_maxev} runs in time $O(nd^2\log^2d)$ and stores at most $O(d^2)$ floats.
\end{theorem}}
\begin{proof}
We first compute the running time and then prove the correctness of the algorithm.

\paragraph{Runtime and Memory Complexity}
It is easy to see that Algorithm \ref{alg:general_greedy_maxev} keeps track of three matrices $U_t, \Sigma_t, V_t$ at every iteration, in addition to the new point $x_t$. The algorithm does not explicitly store the matrices $A_t = U_t \Sigma_t^{-1}V_t$. Instead, when 
it is asked to compute $A_t w$ or $A_t^Tw$, it simply consecutively performs 3 matrix-vector multiplications.
Algorithm \ref{alg:general_greedy_maxev} does not store anything else in between iterations, so the memory complexity is $O(d^2)$ floating-point numbers, as desired.

We now analyze the running time. 
The matrix-vector multiplications require time $O(d^2)$. With the rank-one update, each iteration takes time $O(d^2\log^2d)$. 
Hence, Algorithm \ref{alg:general_greedy_maxev} has time complexity $O(nd^2\log^2d) = \widetilde{O}\inparen{nd^2}$ as desired.  

\paragraph{Correctness}

As in the analysis of Algorithm~\ref{alg:greedy_maxev}, it will in fact be enough to show that for all ellipsoids $\cE^{*}$ that cover $X$ and have aspect ratio at most $\xi$, we have:
\begin{align*}
    \cE_n \subseteq \sqrt{\magicnumber{\xi}} \cdot \cE^{*}
\end{align*} 
Then, by Theorem \ref{thm:simul_ellipsoid_approx}, we will have that the intersection of all ellipsoids $\cE^{*}$ covering $X$ with aspect ratio at most $\xi$ is itself a $\sqrt{2}$-approximation to $X$. Hence, $\cE_n$ will be a $\sqrt{2}\cdot\sqrt{\magicnumber{\xi}}$ to $X$, as required.
Similarly, as in the analysis of Algorithm~\ref{alg:greedy_maxev}, we assume without loss of generality that for all $t$, $\|A_{t-1}x_t\| > 1$. 

Observe that Lemma \ref{lemma:phi_control_max_sv} implies that it is sufficient to analyze $\sigma_{\max}\inparen{\astar \cdot A_{n}^{-1}}$ for our choice of $\astar$ defining the ellipsoid $\cE^{*} = \inbraces{x\suchthat\norm{\astar x}\le 1}$ satisfying $\kappa(\cE^{*})\le\xi$. Specifically, our goal is now to show that $\sigma_{\max}\inparen{\astar \cdot A_{n}^{-1}}^2 \le \magicnumber{\xi}$.

Define $Q_t = \max_{1\le i \le t} \norm{\astar x_i}^2$. Clearly, $Q_1\leq  \dots \leq Q_n$. Now we prove a version of Lemma \ref{lemma:invariant_basic_stronger} that applies to Algorithm \ref{alg:general_greedy_maxev}. We give the proofs of all the lemmas stated below in Appendix~\ref{sec:moved_proofs}.
\begin{lemma}
\label{lemma:invariants}
Under the update rule for Algorithm \ref{alg:general_greedy_maxev}, we have the following

\begin{itemize}
\item For any two timestamps $u$, $t$, such that $1 \leq u \leq t \leq n$:
\begin{equation}
\label{eqn:multi_step_p_inv}
    S_t \leq S_u + Q_t \cdot (P_t - P_u)
\end{equation}
\item For any timestamp \(2 \leq t \leq n\),
\begin{equation}
\label{eqn:single_step_inv}
    S_t \leq S_{t-1} + d \cdot Q_t
\end{equation}
\end{itemize}
\end{lemma}

Let $\sigma_{i,j}$ be the $i$th singular value of $\astar A_{j}^{-1}$ and $\sigma_{\max, j} = \max_i \sigma_{i,j}$.
By Lemma~\ref{lemma:phi_control_max_sv} (part 2), $\alpha \leq \sigma_{\max, n} \leq \sqrt{\sum_{i=1}^d \sigma_{i,n}^2}$, hence upper bounding \(S_n = \sum_{i=1}^d \sigma_{i,n}^2\) is sufficient to finish the analysis.
We do this by bounding $S_1$ and then applying Lemma~\ref{lemma:invariants}, as described below.

\mybox{\begin{lemma}
\label{lemma:s_initial_value}
We have $S_1 \le d$.
\end{lemma}}

By applying the equations from Lemma \ref{lemma:invariants} in sequence, we can relate \(S_n\) to \(S_1\). 
Only applying (\ref{eqn:single_step_inv}) repeatedly is insufficient, as then the resulting upper bound on \(S_n - S_1\) grows linearly with \(n\).
Even though (\ref{eqn:multi_step_p_inv}) is efficient for large stretches where \(t\) is much later than \(u\), it also does not give a good bound when  \(Q_t\) is more than a constant factor larger than \(Q_u\) (this phenomenon may not be apparent from the equation itself, but becomes clear from the upcoming analysis).
In order to obtain the desired bound on $S_n$, we split the \(Q_t\)s into contiguous groups such that \(Q_t\) does not increase significantly within any group, and differs by at least \(e\) outside the groups.

\begin{lemma}\label{lem:bound-on-S_n}
After running Algorithm \ref{alg:general_greedy_maxev} for $n$ steps, we have
\begin{align*}
    S_n = \sum_{i=1}^d \sigma_{i,n}^2 \leq \magicnumber{\xi}
\end{align*}
\end{lemma}

Finally, we write $\sigmamax{n}^2 \le \sum_{i=1}^d \sigma^2_{i,n} \le \magicnumber{\xi}$. Recall that this is sufficient to conclude the proof of Theorem~\ref{thm:greedy_approx} -- specifically, since we have $\sigmamax{n} = \sigma_{\max}\inparen{\astar \cdot A_n^{-1}} \le \sqrt{\magicnumber{\xi}}$, we can invoke part 2 of Lemma~\ref{lemma:phi_control_max_sv} to arrive at $\cE_n \subseteq \sqrt{\magicnumber{\xi}}\cdot\cE^{*}$. 
\end{proof}

\section{Acknowledgments} 
YM and MO were supported in part by NSF awards CCF-1718820, CCF-1955173, and CCF-1934843.
NSM was supported by a United States NSF Graduate Research Fellowship. We thank Meghal Gupta, Darshan Thaker, Taisuke Yasuda, and the anonymous reviewers for helpful feedback and discussions.

\printbibliography
\newpage
\appendix

\section{Proof of Theorem~\ref{thm:simul_ellipsoid_approx}: Approximating Convex Polytopes with Ellipsoids}
\label{section:ellipsoidal_approximations}

In this section, we prove \Cref{thm:simul_ellipsoid_approx}.
To this end, we show that every centrally symmetric convex polytope $X$ is well-approximated by the intersection of all ellipsoids $\cE$ containing $X$ with comparable aspect ratio.
This will immediately imply Theorem~\ref{thm:simul_ellipsoid_approx}.
Let $q \geq \nfrac{1}{\kappa(X)}$ and $\delta = \sqrt{1+\nfrac{1}{q^2}} - 1$. 
Define 
$$
E_{q} = \inbraces{\cE \suchthat X \subseteq \cE \text{ and } \kappa(\cE) \le q\cdot\kappa(X)}
\quad \text{and}\quad
\cA_q = \bigcap_{\cE \in E_q}\cE 
$$
We show that $\cA_q$ provides a good approximation for $X$.
\mybox{\begin{lemma}
\label{thm:ellipsoid_intersections_general}
We have \(
    \frac{1}{1 + \delta}\cdot\cA_q \subseteq X \subseteq \cA_q
\).
\end{lemma}}
\begin{proof}
The first inclusion $X \subseteq \cA_q$ is trivial, since all ellipsoids $\cE$ in $E_q$ contain $X$. We now prove the the first inclusion: 
\begin{equation}\label{eq:AqisinX}
cA_q\subseteq (1+\delta) X.
\end{equation}
Consider the set of all slabs of the form $\{x: |\langle a,x\rangle| \leq 1\}$ that contain $X$.
Since $X$ is a centrally symmetric convex body,
the intersection of all the slabs in $S$ equals $X$.
Further, let set $S'$  consist of the slabs in $S$ expanded by a factor of $(1+\delta)$: for every slab $|\langle a,x\rangle| \leq 1$ in $S$, there is a slab $|\langle a,x\rangle| \leq (1+\delta)$ in $S'$.
Then the intersection of  slabs in $S'$ equals $(1+\delta)X$.
Thus, to prove inclusion~(\ref{eq:AqisinX}), it is sufficient to prove that $|\langle a,x\rangle| \leq (1+\delta)$ in $S$.
To this end, we construct an ellipsoid $\cE$ in $E_q$ that 
lies in the slab $|\langle a,x\rangle| \leq (1+\delta)$
(and contains $\cA_q$, by the definition of $\cA_q$).

We complement vector $\nfrac{a}{\|a\|}$ to an orthonormal basis for $\R^d$:
$w_1 = \nfrac{a}{\|a\|},w_2,\dots,w_d$. 
Recall that the condition number $\kappa(X)$ equals the ratio of the radius $r$ of the largest inscribed ball to the radius $R$ of the smallest circumscribed ball. Since the width of slab $|\langle a,x\rangle| \leq 1$ containing $X$ is $\nfrac{2}{\|a\|}$, we have $r \leq \nfrac{1}{\|a\|}$. Accordingly, $R\leq \nfrac{\kappa(X)}{\|a\|}$.
Therefore, for $x\in X\subseteq R \cdot B_2^d$, 
$\sum_{i=1}^d \langle w_i, x\rangle^2 = \|x\|^2 \leq R^2 \leq 
\inparen{\nfrac{\kappa(X)}{\|a\|}}^2$. Also note that 
$\langle w_1, x\rangle^2 = \nfrac{\ip{a,x}^2}{\norm{a}^2} \leq \nfrac{1}{\norm{a}^2}$. We are ready to define ellipsoid $\cE$:
$$\cE = \left\{x: \langle w_1,x\rangle^2 + \sum_{i=2}^d \frac{\ip{w_i,x}^2}{q^2\kappa(X)^2} \le \frac{1}{\|a\|^2}\Bigl(1 + \frac{1}{q^2}\Bigr)\right\}.$$

Now we verify that (i) $\cE\in E_q$ and (ii) all points in $\cE$ satisfy
$|\langle a,x\rangle| \leq (1+\delta)$.
First, note that the aspect ratio of $\cal E$ is $q\kappa(X)$.
Using the bounds we derived above, we get that for all $x\in X$
$$\langle w_1,x\rangle^2 + \sum_{i=2}^d \frac{\ip{w_i,x}^2}{q^2\kappa(X)^2}
\leq  \frac{1}{\|a\|^2}  + 
\frac{\kappa(X)^2}{\|a\|^2} \cdot \frac{1}{q^2\kappa(X)^2}
= \frac{1}{\|a\|^2}\Bigl(1 + \frac{1}{q^2}\Bigr).
$$
Therefore, all points from $X$ lie in $\cE$; that is, $X\subseteq \cE$. We conclude that $\cE \in E_q$, as required.

Finally, if $x\in \cE$, then 
$$\langle x, a\rangle = \sqrt{\|a\|^2 \langle w_1,x\rangle^2} \leq \sqrt{1 + \nfrac{1}{q^2}} = 1 + \delta.$$
This concludes the proof.
\end{proof}

To get \Cref{thm:simul_ellipsoid_approx}, we apply \Cref{thm:ellipsoid_intersections_general} with $q = 1$.
Clearly, an ellipsoid \(\cE\) that is contained in every \(\cE^{*} \in E_1\) is also contained in \(\cA \subseteq \sqrt{2} \cdot X\).

\section{Tracking the Minimum-Volume Outer Ellipsoid}
\label{section:tracking_mvoe}

Observe that the guarantee of Theorem \ref{thm:greedy_approx} gives a guarantee similar to that given by John's Theorem for centrally-symmetric convex bodies. Therefore, a natural question is, ``how closely can any one-pass monotonic algorithm approximate the minimum-volume outer ellipsoid for a centrally-symmetric convex body?'' We formalize this notion below.

\begin{definition}[Approximation to Minimum Volume Outer Ellipsoid]
We say a streaming algorithm $A$ $\alpha$-approximates the minimum volume outer ellipsoid if $A$ outputs an ellipsoid $\cE_n$ satisfying $\cE_n \subseteq \alpha \cdot J(X)$, where $J(X)$ is the minimum volume outer ellipsoid for $X$.
\end{definition}

Theorem \ref{thm:greedy_approx_tight} asserts that for a natural class of streaming algorithms, it is not possible to approximate the minimum volume outer ellipsoid up to factor $< \sqrt{d}$ in the worst case.

\mybox{\begin{theorem}
\label{thm:greedy_approx_tight}
Every one-pass monotonic deterministic streaming algorithm for Problem \ref{problem:main_problem} has approximation factor to the minimum volume outer ellipsoid of at least $\sqrt{d}$, for infinitely many $d$.
\end{theorem}}
\begin{proof}
Before we begin, recall that a Hadamard basis is a set of vectors $v_1, \dots, v_d$ such that:
\begin{itemize}
    \item $\norm{v_i} = 1$, for all $i \in [d]$;
    \item For all $i \neq j$, $\ip{v_i,v_j}=0$;
    \item Every entry of $v_i$ is in $\inbraces{\pm \nfrac{1}{\sqrt{d}}}$.
\end{itemize}

Our family of hard instances proceeds in two phases.

\greenbox{\paragraph{Phase 1} Let $d$ be such that there exists a Hadamard basis for $\R^d$. Consider a corresponding Hadamard basis $v_1, \dots, v_d$. The adversary gives the algorithm the points $v_1, \dots, v_d$.

\paragraph{Phase 2} The adversary selects $i \in [d]$ arbitrarily and $\eps \in (0, d-1)$ arbitrarily. They then define the vectors $w_i \coloneqq e_i \cdot \nfrac{1}{\sqrt{d-\eps}}$ and $w_j \coloneqq e_j \cdot \sqrt{\nfrac{d-1}{\eps}}$ for all $j \neq i$. The adversary gives the algorithm the points $w_1, \dots, w_d$. Call the outcome here ``Outcome (i).''}

It is easy to see that at the end of Phase 1, the minimum volume outer ellipsoid is simply $B_2^d$. Furthermore, the algorithm's solution $\widehat{\cE}$ contains $\conv{\pm v_1, \dots, \pm v_d}$. On the other hand, consider the following claim.
\mybox{\begin{lemma}
The following ellipsoid is the minimum-volume outer ellipsoid for Outcome (i):
\begin{align*}
    \cE_{OPT(i)} = \inbraces{x \suchthat 1 \ge \frac{x_i^2}{\inparen{\nfrac{1}{\sqrt{d-\eps}}}^2} + \sum_{j \neq i}^d \frac{x_j^2}{\inparen{\sqrt{\nfrac{d-1}{\eps}}}^2}}
\end{align*}
\end{lemma}}
\begin{proof}
Notice that all the points $w_j$ are orthogonal. Thus, the minimum-volume outer ellipsoid containing all the $w_j$ must be the one whose axes are along the directions of $w_j$ and whose poles are located on $w_j$. Observe that $\cE_{OPT(i)}$ satisfies this, so it must be the minimum-volume outer ellipsoid for the convex body whose vertices are determined by the $w_j$.

It now remains to show that every Hadamard basis vector is on the surface of $\cE_{OPT(i)}$:
\begin{align*}
    \frac{\inparen{\nfrac{1}{\sqrt{d}}}^2}{\inparen{\nfrac{1}{\sqrt{d-\eps}}}^2} + \sum_{j\neq i}^d \frac{\inparen{\nfrac{1}{\sqrt{d}}}^2}{\inparen{\sqrt{\nfrac{d-1}{\eps}}}^2} = \frac{1}{d}\inparen{(d-\eps) + (d-1) \cdot \frac{\eps}{d-1}} = 1
\end{align*}
Since the minimum volume ellipsoid containing all the $w_j$ also contains the Hadamard basis vectors, it (i.e., $\cE_{OPT(i)}$) must be the minimum-volume outer ellipsoid for Outcome (i).
\end{proof}

We will now show that any that outputs an ellipsoid $\widehat{\cE}$ at the end of Phase 1 must have an approximation factor of at least $\sqrt{d-\eps}$ on at least one of Outcomes ($1, \dots, i$). Suppose that in each of Outcome (i), we obtain an ellipsoid $\widehat{\cE_i}$ that satisfies $C \cdot \cE_{OPT(i)} \supseteq \widehat{\cE_i}$. We now have:
\begin{align*}
    \conv{\inbraces{\pm v_1, \dots, v_d}} \subseteq \widehat{\cE} \subseteq \bigcap_{i=1}^d \widehat{\cE_{i}} \subseteq C \cdot \bigcap_{i = 1}^d \cE_{OPT(i)}
\end{align*}
We therefore want to argue about $\widehat{\cE}$ given that it must contain $\conv{\inbraces{\pm v_1, \dots, v_d}}$ and be contained by $C \cdot \bigcap_{i = 1}^d \cE_{OPT(i)}$. Let $A$ be a matrix mapping $\widehat{\cE}$ to the unit ball. Then, notice that we can write for all $i \in [d]$:
\begin{align*}
    \norm{Av_i} &\le 1 & \norm{A \cdot \frac{Ce_i}{\sqrt{d-\eps}}} &\ge 1
\end{align*}
In particular, the rightmost exclusion follows from the fact that $\nfrac{Ce_i}{\sqrt{d-\eps}}$ lies on the boundary of $C \cdot \bigcap_{i=1}^d \cE_{OPT(b.i)}$. Now, recall the well-known fact that for any unitary matrix $W$, we have $\fnorm{AW} = \fnorm{A}$ (see, e.g., \cite{horn1991}), and observe that we have:
\begin{align*}
    d \ge \sum_{i=1}^d \norm{Av_i}^2 = \fnorm{AV}^2 = \fnorm{A}^2 = \fnorm{AI}^2 = \sum_{i=1}^d \norm{Ae_i}^2 \ge \frac{d(d-\eps)}{C^2}
\end{align*}
Rearranging gives $C \ge \sqrt{d-\eps}$, as desired.
\end{proof}
\section{Dual Problem -- Inner Ellipsoidal Approximation}
\label{section:mvie}

We design our algorithms for Problem \ref{problem:main_problem} in the setting where we receive points $x_1,\dots, x_n$ defining $X=\conv{\{\pm x_1,\dots, \pm x_n\}}$. Alternatively, we may define a centrally symmetric convex polytope is by providing a set of its faces, or more generally, a set of slabs of the form $\{x:|a^Tx| \leq 1\}$. Accordingly, we may consider a different online model where 
inequalities $\{x:|a^Tx| \leq 1\}$ arrive one-by-one and the resulting polytope is their intersection. Using the notion of a polar set, we show that this model is essentially equivalent to the model we study in this paper. All our results equally apply to it. 

Thus, another possible formulation for Problem \ref{problem:main_problem} involves the algorithm receiving the linear constraints one-at-a-time instead of points from the body.

In this section, we show that this choice of formulation does not matter. Specifically, an algorithm for one of these variants yields an algorithm for the other. In fact, as written in Table \ref{table:primal_vs_dual}, these problems are dual to one another.

\begin{table}[H]
\centering
\begin{tabular}{|l|l|}
\hline
\textbf{Primal}                                                                                                                                                               & \textbf{Dual}                                                                                                                                                                                                   \\ \hline
\begin{tabular}[c]{@{}l@{}}Find $\cE$ (outer ellipsoid) such that:\\  $\nfrac{\cE}{\alpha} \subseteq X \subseteq \cE$\\ where $X = \conv{\inbraces{\pm x_1,\dots,\pm x_n}}$\end{tabular} & \begin{tabular}[c]{@{}l@{}}Find $\cE$ (inner ellipsoid) such that:\\  $\cE \subseteq Y \subseteq \alpha\cdot\cE$\\ where $Y = \inbraces{y \suchthat \abs{\ip{x_i,y}}\le 1\text{ for all } i\in[n]}$\end{tabular} \\ \hline
\end{tabular}
\caption{The primal and dual version of the ellipsoidal approximation problem.\label{table:primal_vs_dual}}
\end{table}

We first address the duality between the two problems in Table \ref{table:primal_vs_dual}. To do so, observe the following useful facts regarding convex polars.
\begin{itemize}
    \item If $A$ and $B$ are convex bodies, and if $A \subseteq B$, then $B^{\circ} \subseteq A^{\circ}$ (see Proposition 7.16(iv) in \cite{convexanalysis}).
    \item In Table \ref{table:primal_vs_dual}, $X^{\circ} = Y$, and $Y^{\circ} = X$.
    \item If an ellipsoid $\cE = \inbraces{x \suchthat \norm{Ax}\le 1}$, then $\cE^{\circ} = \inbraces{x \suchthat \norm{A^{-T}x} \le 1}$ (see Definition 2.17 in \cite{vishnoi2021}).
\end{itemize}
The first and third are well-known, and the second follows from the definition of the polar and that if $X$ is closed, convex, and contains the origin, then $(X^{\circ})^{\circ} = X$ (see Corollary 7.19(i) in \cite{convexanalysis}).

We now put these facts together to show that a solution to the primal problem can be converted to one for the dual problem. First, notice that we have $\cE$ and $\alpha$ such that $\nfrac{\cE}{\alpha} \subseteq X \subseteq \cE$. Using the first fact, we have $\cE^{\circ} \subseteq X^{\circ} \subseteq \inparen{\nfrac{\cE}{\alpha}}^{\circ}$. Using the second fact, we have $\cE^{\circ} \subseteq Y \subseteq \inparen{\nfrac{\cE}{\alpha}}^{\circ}$. Finally, using the third fact, we have $\inparen{\nfrac{\cE}{\alpha}}^{\circ} = \alpha \cdot \cE^{\circ}$, which yields $\cE^{\circ} \subseteq Y \subseteq \alpha\cdot\cE^{\circ}$. A similar argument shows that a solution to the dual yields a solution to the primal.

We now address how a solution to the streaming variant of the primal problem can be converted to a solution to the streaming variant of the dual problem. Specifically, suppose we are in the dual setting, wherein we receive linear constraints one-at-a-time. Our task is to find an $\alpha$-ellipsoidal approximation to $Y_t = \inbraces{y \suchthat \abs{\ip{x_i,y}}\le 1 \text{ for all } i \in [t]}$. Observe that every incoming linear constraint $\inbraces{y \suchthat \abs{\ip{x_i, y}} \le 1}$ can be treated as an incoming point $\pm x_i$ in the primal space. This means that we can apply our algorithm in the primal setting to obtain a solution in the primal space, which for all $t$ gives an ellipsoid such that $X_t \subseteq \cE_t$. We then compute the polar of the outer ellipsoid we obtain in the primal space (i.e., $\cE_t$) to obtain an inner ellipsoid in the dual space (i.e., $\cE_t^{\circ}$), which yields $\cE_t^{\circ} \subseteq Y_t$. As per our previous argument, this preserves the approximation factor -- at the end of the stream, we have $\cE_n^{\circ} \subseteq Y \subseteq \alpha \cdot \cE_n^{\circ}$, as desired.
\section{Proof of Claim~\ref{claim:minimum_volume_ellipsoid_update}}
\label{sec:rank-one-update}
Note that the volume of the ellipsoid determined by $A_t$ is proportional to $\detv{A_t^{-1}}$. Therefore, $A_t$ is the solution to the following optimization problem, where we use that the volume of the ellipsoid determined by $A_t$ is proportional to $\detv{A_t^{-1}}$.
\begin{align*}
    \max \detv{A_t} \text{ such that } A_t \preceq A_{t-1} \text{ and } \norm{A_tx_t} \le 1
\end{align*}
Additionally, since $\detv{AB} = \detv{A}\cdot\detv{B}$, we have that this objective is invariant under linear transformations. It thus follows that our objective can be rewritten as:
\begin{align*}
    &\max \detv{A_t} &\text{ such that } &A_t \preceq A_{t-1} \text{ and } \norm{A_tx_t} \le 1 \\
    \equiv &\max \detv{A_t \cdot A_{t-1}^{-1}} &\text{ such that } &A_t \cdot A_{t-1}^{-1} \preceq I \text{ and } \norm{\inparen{A_t\cdot A_{t-1}^{-1}}A_{t-1}x_t} \le 1 \\
    \equiv &\max \detv{\widehat{A}} &\text{ such that } &\widehat{A} \preceq I \text{ and } \norm{\widehat{A}A_{t-1}x_{t}} \le 1
\end{align*}
where the last line follows from using the intermediate variable $\widehat{A} = A_t\cdot A_{t-1}^{-1}$.

In other words, after the transformation, the problem is equivalent to finding the minimum volume ellipsoid that contains (i) the unit ball and (ii) point $A_{t-1}x_t$. Geometrically, it is clear what the optimal ellipsoid for this problem is: one of its semi-axes is $A_{t-1}x_t$; all others are orthogonal to $A_{t-1}x_t$ and have length 1 (this can be formally proved using symmetrization). However, we do not use this observation and derive a formula for $\widehat{A}$ using linear algebra.

We first give an upper bound on the objective value of the above optimization problem. Since $\widehat{A} \preceq I$, we have that all its singular values must be at most $1$. Additionally, since $1 \ge \norm{\widehat{A}A_{t-1}x_t} \ge \sigma_{\min}\inparen{\widehat{A}} \cdot \norm{A_{t-1}x_t}$, we have that at least one singular value of $\widehat{A}$ must be $\le \nfrac{1}{\norm{A_{t-1}x_t}}$. Putting everything together and using the fact that the determinant is the product of the singular values gives $\detv{\widehat{A}} \le \nfrac{1}{\norm{A_{t-1}x_t}}$.

We now show that there exists a setting of $\widehat{A}$ that achieves this upper bound. Let $v_1$ be a unit vector in the direction of $A_{t-1}x_t$ and $v_2, \dots, v_d$ complete the orthonormal basis for $\R^d$ from $v_1$, and write \(
    \widehat{A} = \frac{1}{\norm{A_{t-1}x_t}}v_1v_1^T + \sum_{i=2}^d v_iv_i^T
\). We will show that $\widehat{A}$ satisfies the constraints imposed by the optimization problem. Since we have $\norm{A_{t-1}x_t} \ge 1$ (as we impose that $x_t \notin \cE_{A_{t-1}}$), the fact that $\widehat{A} \preceq I$ follows immediately. For the second constraint, we write:
\begin{align*}
    \norm{\widehat{A}A_{t-1}x} = \norm{\inparen{\frac{1}{\norm{A_{t-1}x_t}}v_1v_1^T + \sum_{i=2}^d v_iv_i^T}A_{t-1}x_t} = \norm{\frac{A_{t-1}x_t}{\norm{A_{t-1}x_t}}} = 1
\end{align*}
Furthermore, it is easy to see that $\detv{\widehat{A}} = \nfrac{1}{\norm{A_{t-1}x_t}}$, which achieves our upper bound.

Finally, recall that we wrote $\widehat{A} = A_t\cdot A_{t-1}^{-1}$; rearranging this gives us what we want.

\section{Proofs from Section~\ref{section:revised_algorithm}}
\label{sec:moved_proofs}
For the purposes of our analysis, we will ``simulate'' the singular value correction step using the following procedure. Let $w_{1}, \dots, w_{d}$ be the the $i$-th column of $V_t$ (note that the $w_i$-s
are unit vectors that are the directions of the semi-axes of $\cE_t$).
Let $\cE_t'$ be the ellipsoid obtained prior to Line \ref{line:sv_correction}.
i.e., $\cE_t' = \inbraces{x\suchthat\norm{(\Sigma_t')^{-1}V_t^Tx}\leq 1}$.
We create ``ghost'' points $z_{1}, \dots, z_{d}$: $z_{i} = \tau_{i,t} w_i = \max([\Sigma'_{t}]_{i i}, 
\nfrac{M_t}{\xi}) w_i $ (see line \ref{line:sv_correction} in Algorithm \ref{alg:general_greedy_maxev}).
Note that $X$ contains the ball of radius $M_t/\xi$ centered at $0$, since the aspect ratio of $X$ is at most $\xi$. Thus, each point $z_i$ either lies in 
$X$ (if $\|z_i\| = \nfrac{M_t}{\xi}$) or 
in $\cE_t'$ (if $\|z_i\| = [\Sigma'_{t}]_{i i}$).
We finally start with matrix $A_{t-1}$ and consecutively apply the update rule from Algorithm~\ref{alg:greedy_maxev} for each of the points $x_t, z_{1}, \dots, z_{d}$. 

Next, we show Lemma~\ref{lemma:sv_update_step_simulation}, which states that this simulation of the singular value correction step yields $A_t$, the same matrix that we obtain when we perform the update rule from Algorithm~\ref{alg:general_greedy_maxev}.

\begin{lemma}
\label{lemma:sv_update_step_simulation}
The process described above yields matrix $A_t$.
\end{lemma}
\begin{proof}
Consider the executions of Algorithms \ref{alg:greedy_maxev} and \ref{alg:general_greedy_maxev}.
After Algorithm~\ref{alg:greedy_maxev} processes point $x_t$ and Algorithm \ref{alg:general_greedy_maxev} executes Line \ref{line:svd_update}, both algorithms are in the same state. Namely they store matrix $A_{t}'$ given by Claim~\ref{claim:minimum_volume_ellipsoid_update}, 
\begin{align*}
    A_t' = A_{t-1} - \inparen{1 - \frac{1}{\norm{A_{t-1}x_t}}}\inparen{\frac{\inparen{A_{t-1}x_t}\inparen{A_{t-1}x_t}^T}{\norm{A_{t-1}x_t}^2}}A_{t-1}
\end{align*}
It remains to show that executing Line \ref{line:sv_correction} in Algorithm \ref{alg:general_greedy_maxev} is equivalent to injecting these ``ghost'' points $z_1,\dots,z_d$ into Algorithm \ref{alg:greedy_maxev}.

It is sufficient to consider the effect of injecting one point $z_i$. Assume we started with matrix $A$ and obtained matrix $A'$ by injecting $z_i$. Observe that if $\sigma_i(A) \leq \nfrac{1}{\|z_i\|}$ (that is, $z_i$ lies in the ellipsoid defined by $A$; in particular, if $z_i \in \cE_t'$), then $A'=A$.
We prove that if $\sigma_i(A) > \nfrac{1}{\|z_i\|}$, then $A' = U{\Sigma'}^{-1}V^T$ where $A = U\Sigma^{-1}V^T$, entry $\Sigma_{ii}' = \|z_i\|$, and all other entries of $\Sigma'$ are equal to the corresponding entries of $\Sigma$. We have,
$$
    A' = \widehat{A} A = \left(I - \inparen{1 - \frac{1}{\norm{A z_i}}}\inparen{\frac{\inparen{Az_i}\inparen{Az_i}^T}{\norm{Az_i}^2}}\right)A 
$$
Since $z_i = \tau_{i,t} w_{i}$, we have $V^T z_i = \|z_i\| e_i = \tau_{i,t} e_i$. Accordingly, 
$Az_i = U\Sigma^{-1} \tau_{i,t} e_i =  \frac{\tau_{i,t}}{\Sigma_{ii}} Ue_i$.
Thus, $\widehat{A} = I - (1 - \frac{\Sigma_{ii}}{\tau_{i,t}}) Ue_i e_i^T U^T$.
Since $U$ is an orthogonal matrix, $UU^T = I$ and thus 
$$A'=\left(I - \inparen{1 - \frac{\Sigma_{ii}}{\tau_{i,t}}} Ue_i e_i^T U^T\right)U\Sigma^{-1}V^T = 
U \biggl(\underbrace{\Sigma^{-1} \left(I-\inparen{1 - \frac{\Sigma_{ii}}{\tau_{i,t}}} e_ie_i^T\right)}_{\Sigma'^{-1}}\biggr)V^T.$$
Note that $I-\bigl(1 - \frac{\Sigma_{ii}}{\tau_{i,t}}\bigr) e_ie_i^T$ is a diagonal matrix; all of its diagonal entries are equal to 1 except for the $i$-th diagonal entry, which is $\nfrac{\Sigma_{ii}}{\tau_{i,t}}$.
We get that $\Sigma'^{-1}$ differs from $\Sigma^{-1}$ only in the $i$-th diagonal entry: $(\Sigma'^{-1})_{ii} = \nfrac{1}{\tau_{it}}$, as required.
\end{proof}

\begin{proof}[Proof of Lemma~\ref{lemma:invariants}]
By Lemma~\ref{lemma:sv_update_step_simulation}, we can break the evolution of our potential functions into two main steps: that after Algorithm~\ref{alg:greedy_maxev} gets $x_t$ and that after Algorithm~\ref{alg:greedy_maxev} gets the ghost points.
Let $S_t'$, $P_t'$, $Q_t$ represent $\fnorm{\astar\cdot (A_t')^{-1}}^2$, $2\log\detv{(A_t')^{-1}}$, and $\max_{i\in[t]} \norm{\astar x_i}^2$, respectively, where $A_t'$ is the matrix defined in Lemma~\ref{lemma:sv_update_step_simulation}.

By Lemma \ref{lemma:invariant_basic_stronger}, for all $t$, we have $S_t' - S_{t-1} \le (P_t' - P_{t-1}) \cdot Q_t$. Similarly, as \(\|\astar x_t\| \leq 1\) and \(\|A_{t-1} x_t\| > 1\), we have:
\begin{align*}
    S_t' - S_{t-1} = \inparen{1 - \frac{1}{\norm{A_{t-1}x_t}^2}}\norm{\astar x_t}^2 \le \inparen{1 - \frac{1}{\norm{A_{t-1}x_t}^2}}Q_t \le Q_t
\end{align*}
We now analyze the singular value correction step. By Lemma~\ref{lemma:sv_update_step_simulation},it can be simulated by adding the ghost points $z_1,\dots,z_d$.
First, observe that we only need to analyze points $z_i$ with  $z_i = \nfrac{M_t}{\xi} w_{i}$ (because other points are in $\cE_t'$ and do not cause any update). We now show that $\norm{\astar z_i} \le \max_{i \in [t]} \norm{\astar x_t} = \sqrt{Q_t}$. We have:
\begin{align*}
    \norm{\astar z_i} = \frac{M_t}{\xi} \cdot \norm{\astar w_i} \le M_t \cdot \frac{\sigma_{\max}\inparen{\astar}}{\xi} \le \sigma_{\min}\inparen{\astar} \cdot M_t
\end{align*}
Let $p \coloneqq \argmax_{i \in [t]} \norm{x_i}$. Then:
\begin{align*}
    \norm{\astar z_i} \le \sigma_{\min}\inparen{\astar} \cdot M_t = \sigma_{\min}\inparen{\astar} \cdot \norm{x_p} \le \norm{\astar x_p} \le \max_{i \in [t]} \norm{\astar x_t} = \sqrt{Q_t}.
\end{align*}
Consider the state of Algorithm~1 in the simulation after it gets points $z_1,\dots,z_i$.
Let $A_{t,i}'$ be the resulting state matrix, and $S_{t,i}'$ and $P_{t,i}'$ be the values of functions $S$ and $P$. Observe that $S_{t,0}' = S_{t}'$ and $P_{t,0}' = P_{t}'$ and that $S_t = S_{t,d}'$ and $P_t = P_{t,d}'$. We now invoke Lemma~\ref{lemma:invariant_basic_stronger} repeatedly:
\begin{align*}
    S_{t,0}' - S_{t-1} &\le (P_{t,0}' - P_{t-1}) \cdot Q_t \\
    S_{t,1}' - S_{t,0}' &\le (P_{t,1}' - P_{t,0}') \cdot Q_t \\
    &\vdots \\
    S_{t,i}' - S_{t,i-1}' &\le \inparen{P_{t,i}' - P_{t,i-1}'} \cdot Q_t \\
    &\vdots \\
    S_{t} - S_{t,d-1}' &\le \inparen{P_{t} - P_{t,d-1}'} \cdot Q_t
\end{align*}
Adding all these inequalities yields $S_{t}-S_{t-1} \le \inparen{P_t-P_{t-1}}\cdot Q_t$. Since numbers $Q_t$ are non-decreasing, this implies that for all $u \le t$, we have $S_t \leq S_u + Q_t \cdot (P_t - P_u)$.

Similarly, we repeatedly write:
\begin{align*}
    S_{t,0}' - S_{t-1} &\le Q_t \\
    S_{t,1}' - S_{t,0}' &\le Q_t \\
    &\vdots \\
    S_{t,i}' - S_{t,i-1}' &\le Q_t \\
    &\vdots \\
    S_{t} - S_{t,d-1}' &\le Q_t
\end{align*}
Note that at least one of the semi-axes of $\cE_t'$ has length at least $M_t$
(since all points $x_1,\dots, x_t$ are in $\cE_t'$). That is, $[\Sigma']_{ii} \geq M_t \geq M_t/\xi_i$ for some point $z_i$.
This means that we do not perform any updates for $z_i$ and $S_{t,i}' = S_{t,i-1}'$. Summing up the inequalities above and taking into account that $S_{t,i}' = S_{t,i-1}'$ for at least one $i$ yields $ S_t \leq S_{t-1} + d \cdot Q_t$, as required.
\end{proof}

\begin{proof}[Proof of Lemma~\ref{lemma:s_initial_value}]
Consider ellipsoid $\cE_1$. It has semi-axes $\tilde w_{1}=x_1$ and some $\tilde w_{2},\dots, \tilde w_{d}$. Since $\|\tilde w_{i}\| = \nfrac{\|x_1\|}{\xi}$ for $i \geq 2$ (see step 5 of Algorithm~\ref{alg:general_greedy_maxev}), all points $\tilde w_{i}$ lie inside $X\subset \cE^{*}$. In particular, $\|\astar \tilde w_{i}\| \leq 1$. Finally, note that $A_1 \tilde w_{i} = e_i$.
We have:
\begin{align*}
   S_1 = \|\astar A_1^{-1}\|_F = \sum_{i=1}^d \|\astar A_1^{-1} e_i\|^2 = 
\sum_{i=1}^d \|\astar \tilde w_i\|^2 \leq d 
\end{align*}
\end{proof}

\begin{proof}[Proof of Lemma~\ref{lem:bound-on-S_n}]
Set $t_0 = n$, and define $t_i \in \{1, \ldots, n\}$ for $i \geq 1$ recursively as follows:
\begin{enumerate}
    \item Set \(t_i = \max \{j \in \{1, \ldots, n\} \colon Q_j < \nfrac{Q_{t_{i-1}}}{e} \}\)
    \item If there is no such \(j\), then finish.
\end{enumerate}
Let \(t_0 > \ldots > t_m\) be the indices defined by this process. Observe that for each \(0 \leq i \leq m\), we have \(Q_{t_i} \leq e^{-i}\). Further, for each \(0 \leq i \leq m -1\), we have \(\nfrac{Q_{t_i}}{Q_{t_{i+1} +1}} \leq e\), and \(\nfrac{Q_{t_m}}{Q_1} \leq e\).

Now we apply bound (\ref{eqn:multi_step_p_inv}). For every \(0 \leq i \leq m - 1\), we have
$$S_{t_i} \leq S_{t_{i+1} + 1} + Q_{t_i} (P_{t_i} - P_{t_{i+1} + 1})$$
From bound (\ref{eqn:single_step_inv}), we get 
$$S_{t_{i+1} + 1} \leq S_{t_{i+1}} + d \cdot Q_{t_{i+1} + 1}.$$
Combining these equations, we obtain for each \(0 \leq i \leq m - 1\):
\begin{equation}
\label{eqn:combined_inv}
S_{t_i} - S_{t_{i+1}} \leq d \cdot Q_{t_{i+1} + 1} + Q_{t_i} \cdot (P_{t_i} - P_{t_{i+1} + 1})
\end{equation}
Now we add up inequalities (\ref{eqn:combined_inv}) for all \(0 \leq i \leq m-1\) and get:
\[S_{t_m} - S_{t_0} = \sum_{i=0}^{m-1} \left(S_{t_i} - S_{t_{i+1}}\right) \leq d \cdot \sum_{i=0}^{m-1} Q_{t_{i+1} + 1} + \sum_{i=0}^{m-1} Q_{t_{i}} \cdot P_{t_i} + \sum_{i=0}^{m-1} Q_{t_{i}} \cdot \left(- P_{t_{i+1} + 1}\right)\]
By (\ref{eqn:multi_step_p_inv}) we have \(S_{t_m} - S_1 \leq Q_{t_m} \cdot (P_{t_m} - P_1)\). Combining this with the previous inequality, we get
\begin{equation}
\label{eqn:final_bound}
S_{n} - S_1 \leq d \cdot \underbrace{\sum_{i=0}^{m-1} Q_{t_{i+1} + 1}}_{\text{I}} + \underbrace{\sum_{i=0}^{m} Q_{t_{i}} \cdot P_{t_i}}_{\text{II}} + \underbrace{\left(-Q_{t_m} \cdot P_1\right) + \sum_{i=0}^{m-1} Q_{t_{i}} \cdot \left(- P_{t_{i+1} + 1}\right)}_{\text{III}}
\end{equation}
Now, we bound each individual term in (\ref{eqn:final_bound}).
For I, observe that as \(t_{i+1} + 1 \leq t_i\), we have \(Q_{t_{i+1}+1} \leq Q_{t_i} \leq e^{-i}\). Thus
\[\sum_{i=0}^{m-1} Q_{t_{i+1} + 1} \leq \sum_{i=0}^{m-1} e^{-i} \leq \sum_{i=0}^\infty e^{-i} = \frac{e}{e-1}\]
For II, we have \(P_{t_i} \leq P_{t_0}\) for all \(i\), therefore
\[\sum_{i=0}^m Q_{t_i} P_{t_i} \leq P_{t_0} \cdot \sum_{i=0}^m Q_{t_i} \leq P_{t_0} \cdot \frac{e}{e-1}\]

Now we bound III. First, we prove the following lemma to bound terms of the form \(- Q_t \cdot P_u\):
\mybox{
\begin{lemma}
\label{lemma:prod_qp_bound}
Consider indices \(u \leq t\). If \(\nfrac{Q_{t}}{Q_u} \leq e\), then
\[- Q_t \cdot P_u \leq Q_t \cdot \left(d + 4d \log \xi + d \log\left(\frac{1}{Q_t}\right)\right)\]
\end{lemma}
}
\begin{proof}
By definition, we have \(-P_u = - \sum_{i=1}^d \log(\sigma^2_{i, u})\). Next, notice that for all $t\in\{1,\dots,u\}$:
\begin{align*}
    \sigma_{\min}\inparen{\astar \cdot A_{u}^{-1}} &\ge \sigma_{\min}\inparen{\astar} \cdot \sigma_{\min}\inparen{A_{u}^{-1}} \\
    &\ge \sigma_{\min}\inparen{\astar} \cdot \frac{\norm{x_t}}{\xi}\quad\quad\text{since the correction step ensures that } \sigma_{\min}\inparen{A_u^{-1}} \ge \frac{\norm{x_t}}{\xi} \\
    &\ge \frac{\sigma_{\max}\inparen{\astar}}{\xi}\cdot \frac{\norm{x_t}}{\xi} \\
    &\ge \frac{\norm{\astar x_t}}{\xi^2}
\end{align*}
Since this is true for all $t \in [u]$, we can maximize the RHS over $t \in [u]$, and we obtain \(\sigma_{i, u}^2 \geq \nfrac{Q_u}{\xi^4}\). Hence, \(-P_u \leq d \logv{\nfrac{1}{Q_u}} + 4d \log \xi\). As \(\nfrac{1}{Q_u} \leq \nfrac{e}{Q_t}\), we get \(-P_u \leq d + 4d \log \xi + d \logv{\nfrac{1}{Q_t}}\). Multiplying by \(Q_t\) gives the claim.
\end{proof}

Applying Lemma \ref{lemma:prod_qp_bound} to each term in III, we obtain
\[-Q_{t_m} \cdot P_1 + \sum_{i=0}^{m-1} Q_{t_{i}} \cdot - P_{t_{i+1} + 1} \leq (d + 4d \log \xi) \cdot \sum_{i=0}^m Q_{t_i} + d \cdot \sum_{i=0}^m Q_{t_i} \log\left(\frac{1}{Q_{t_i}}\right)\]
As before, we can bound the first term with \(\sum_{i=0}^m Q_{t_i} \leq \frac{e}{e-1}\).
For the second term, observe that \(y \mapsto y \log(\nfrac{1}{y})\) is increasing on \(\insquare{0,\nfrac{1}{e}}\). Thus, \(Q_{t_i} \log(\nfrac{1}{Q_{t_i}}) \leq e^{-i} \log(\nfrac{1}{e^{-i}}) = e^{-i} \cdot i\) for \(i \geq 1\). The maximum of \(y \logv{\nfrac{1}{y}})\) is \(\nfrac{1}{e}\); therefore \(Q_{t_0} \log(\nfrac{1}{Q_{t_0}}) \leq \nfrac{1}{e}\).
Thus we can bound the second term by 
$$d \cdot \left(\frac{1}{e} + \sum_{i=1}^{m} i \cdot e^{-i}\right) \leq d \cdot \left(\frac{1}{e} + \sum_{i=1}^{\infty} i \cdot e^{-i}\right) = d \cdot \inparen{\frac{1}{e} + \frac{e}{(1-e)^2}}.$$
To summarize, we bound term III by
\[d \cdot \left(\frac{1}{e} + \frac{e}{e-1} + \frac{e}{(1-e)^2}\right) + 4d \log \xi \cdot \frac{e}{e-1}\]

Combining the bounds on the terms of (\ref{eqn:final_bound}) results in:
\[
    S_n - S_1 \leq d \cdot \left(\frac{1}{e} + \frac{2 \cdot e}{e-1} + \frac{e}{(1-e)^2}\right) + P_{t_0} \cdot \frac{e}{e-1} + 2d \log \xi \cdot \frac{e}{e-1}
\]
Rearranging and applying \Cref{lemma:s_initial_value}, we obtain
\begin{equation}
\label{eqn:final_sn_pn}
    S_n -  P_{n} \cdot \frac{e}{e-1} \leq  d \cdot \left(1 + \frac{1}{e} + \frac{2 \cdot e}{e-1} + \frac{e}{(1-e)^2}\right) + 2d \log \xi \cdot \frac{e}{e-1}
\end{equation}
Observe that
\(S_n - P_n \cdot \frac{e}{e-1} = \sum_{i=1}^d (\sigma^2_{i, n} - \frac{e}{e-1} \log(\sigma^2_{i, n}))\). As \(y - \frac{e}{e-1}\log y \geq \frac{e-2}{e-1}\, y\) for all \(y > 0\), we get
\[
\sum_{i=1}^d \sigma^2_{i, n} \leq \frac{e-1}{e-2} \left(S_n - P_n \cdot \frac{e}{e-1}\right)
\]

Using (\ref{eqn:final_sn_pn}) and replacing constants with their integer ceilings, we finish with
\[\sum_{i=1}^d \sigma^2_{i,n} \leq \magicnumber{\xi}\]
\end{proof}
\end{document}